\documentclass[a4paper,UKenglish,cleveref, autoref, thm-restate]{lipics-v2021}
\nolinenumbers
\title{\greedyl: A Parallel Algorithm for Maximizing Constrained Submodular Functions} %TODO Please add

\author{Shivaram Gopal}{Purdue University}{}{}{}

\author{SM Ferdous}{Pacific Northwest National Lab}{}{}{}

\author{Hemanta Maji}{Purdue University}{}{}{}

\author{Alex Pothen}{Purdue University}{}{}{}

\authorrunning{Shivaram, SM, Alex and Hemanta} %TODO mandatory. First: Use abbreviated first/middle names. Second (only in severe cases): Use first author plus 'et al.'

\Copyright{Jane Open Access and Joan R. Public} %TODO mandatory, please use full first names. LIPIcs license is "CC-BY";  http://creativecommons.org/licenses/by/3.0/


 %TODO mandatory; please add comma-separated list of keywords

\category{} %optional, e.g. invited paper

\relatedversion{} %optional, e.g. full version hosted on arXiv, HAL, or other respository/website
\EventEditors{John Q. Open and Joan R. Access}
\EventNoEds{2}
\EventLongTitle{42nd Conference on Very Important Topics (CVIT 2016)}
\EventShortTitle{CVIT 2016}
\EventAcronym{CVIT}
\EventYear{2016}
\EventDate{December 24--27, 2016}
\EventLocation{Little Whinging, United Kingdom}
\EventLogo{}
\SeriesVolume{42}
\ArticleNo{23}
\usepackage{amsfonts, amssymb, mathtools,amsmath}

\usepackage{algorithmicx, algpseudocode, algorithm}

\usepackage{relsize}

\usepackage[normalem]{ulem}

\usepackage{multirow}
\usepackage{chngpage} % allows for temporary adjustment of side margins
\usepackage{array}
\usepackage{booktabs}

\usepackage{mfirstuc}

\usepackage{tikz}
\usepackage{graphicx}
\usetikzlibrary{arrows}
\usepackage{forest}
\usetikzlibrary{calc,decorations.pathmorphing,patterns,positioning}
\usetikzlibrary{shapes.geometric, shapes.misc, arrows.meta, positioning, decorations.pathreplacing, backgrounds, fit, matrix}
\usepackage[font=tiny]{subcaption}
\usepackage{comment} 

\usepackage[colorinlistoftodos,prependcaption,textsize=footnotesize]{todonotes}
\usepackage{xspace}
\usepackage{xargs}  
\usepackage{multicol}
\usepackage[c2, nocomma, short]{optidef}
\usepackage{nicefrac}

\usepackage{xurl}
\usepackage{enumitem}

\let\epsilon\varepsilon
\let\leq\leqslant
\let\geq\geqslant
\newcommand\restr[2]{{% we make the whole thing an ordinary symbol
  \left.\kern-\nulldelimiterspace % automatically resize the bar with \right
  #1 % the function
  \vphantom{\big|} % pretend it's a little taller at normal size
  \right|_{#2} % this is the delimiter
  }}

\newcommand{\pdfsamp}{dual\xspace}
\newcommand{\cdfsamp}{cumulative dual\xspace}
\newcommand{\Pdfsamp}{\expandafter\capitalisewords\expandafter{\pdfsamp}}
\newcommand{\Cdfsamp}{\expandafter\capitalisewords\expandafter{\cdfsamp}}

\makeatletter

\newcommand{\Rom}[1]{\expandafter\@slowromancap\romannumeral #1@}

\makeatother

\newcommand{\fhat}{\widehat{f}}

\newcommand{\lsz}{Lov\'{a}sz }
\newcommand{\greedyl}{\textsc{GreedyML}\xspace}
\newcommand{\greedy}{\textsc{Greedy}\xspace}
\newcommand{\greedi}{\textsc{GreeDi}\xspace}
\def\greedyl{\mbox{\sc GreedyML}}
\def\randgreedi{\mbox{\sc RandGreedi}}
\def\greedi{\mbox{\sc Greedi}}

\makeatletter
\newcommand{\pushright}[1]{\ifmeasuring@#1\else\omit\hfill$\displaystyle#1$\fi\ignorespaces}
\newcommand{\pushleft}[1]{\ifmeasuring@#1\else\omit$\displaystyle#1$\hfill\fi\ignorespaces}
\makeatother

\providecommand{\ceil}[1]{\ensuremath{\left\lceil {#1} \right\rceil}\xspace}
\makeatletter 
\@mparswitchfalse%
\makeatother
\normalmarginpar %for right-handed notes and lines, or 
\newcommandx{\smf}[2][1=]{\todo[color=cyan!50,#1]{\sf \textbf{smf:} #2}\xspace}

\newcommandx{\sg}[2][1=]{\todo[color=cyan!50,#1]{\sf \textbf{SG:} #2}\xspace}

\newcommand{\customhead}[1]{\smallskip\noindent\textbf{#1.}}

\begin{document}
% Anonymized version  for blind review

%\centerline{\Large \bf A Parallel Algorithm for Maximizing  Submodular Functions} 

\maketitle

\begin{abstract}
We describe a parallel approximation algorithm for maximizing monotone submodular functions subject to hereditary constraints on distributed memory multiprocessors.
Our work is motivated by the need to solve submodular optimization problems on massive data sets, for practical contexts such as data summarization,  machine learning, and graph sparsification. 

Our work builds on the randomized distributed \randgreedi\ algorithm, proposed by  Barbosa, Ene, Nguyen, and Ward (2015). 
This algorithm computes a distributed solution by randomly partitioning the data among all the processors and then employing  \emph{a single} accumulation step in which all processors send their partial solutions to one processor. 
However, for large problems, the accumulation step  exceeds the memory available on a processor, 
and the processor which performs the accumulation
becomes a computational bottleneck. 

Hence we propose a generalization of the \randgreedi\ algorithm that employs multiple accumulation steps to reduce the memory required. 
We analyze the approximation ratio and the time complexity of the algorithm (in the BSP model).
We evaluate the new \greedyl\ algorithm on three classes of problems, and report results from large-scale data sets with millions of elements.  
The results show that the \greedyl\ algorithm can solve problems where the sequential \greedy\ and distributed \randgreedi\ algorithms fail due to memory constraints.    
For  certain computationally intensive problems, the \greedyl\ algorithm is faster than the  
\randgreedi\ algorithm. 
The observed approximation quality of the solutions computed by the \greedyl \ algorithm closely matches those obtained by the \randgreedi\  algorithm on these problems.  

\end{abstract}

%\noindent {\bf Keywords}: Combinatorial  optimization, submodular functions, parallel algorithms, 
%\newline 
%approximation algorithms, data summarization. 
%\newline 
%\noindent {\bf Mathematics Subject Classification}: 90C27, 68W10. 

\section{Introduction}
\begin{comment}
Shivaram: Goals of the introduction. Explain these
1. where and why submodular is used? bring machine learning context.
2. why do we need a distributed algorithm?
3. what is the current state of the art? what are its drawbacks?
4. what is our algorithm? In what ways are be better?
    a. better memory usage
    b. better execution time
    c. better communication
\end{comment}
We describe \greedyl{}, a parallel approximation algorithm for maximizing monotone submodular functions subject to hereditary constraints on distributed memory multiprocessors. \greedyl{} is built on an earlier distributed approximation algorithm, which has limited parallelism and higher memory requirements.  
Maximizing a submodular function under constraints is NP-hard, but a natural iterative \greedy{} algorithm exists  that selects elements based on the marginal gain (defined later) and is $(1 - 1/\mathrm{e}) \approx 0.63$-approximate  for cardinality constraints and 
$1/2$-approximate  for matroid constraints; here $\mathrm{e}$ is Euler's number. 

Maximizing a submodular function (rather than a linear objective function)  promotes diversity in the computed solution since at each step the algorithm encourages an element with the least properties in common with the current solution set. A broad collection of practical problems are modeled using submodular functions, including data and document summarization~\cite{Mirzasoleiman}, load balancing parallel computations in quantum chemistry~\cite{Ferdous+:ACDA21}, sensor placement~\cite{Coutino+:sensing},   resource allocation~\cite{Thekumparampil+:resource},  
active learning~\cite{Golovin+:learning}, interpretability of neural networks~\cite{Elenberg+:NNs}, influence maximization in social networks~\cite{kempe2003maximizing}, diverse recommendation~\cite{chen2018fast}. 
%Submodular optimization problems often have efficient approximation algorithms since they are shown to be discrete analogs of both convex and concave continuous functions. 
Surveys discussing submodular optimization formulations, algorithms, and computational experiments include Tohidi et al.~\cite{Tohidi+:survey} and 
Krause and Golovin~\cite{Krause+:submodularity}. 

Our algorithm builds on the \randgreedi\  framework~\cite{barbosa}, a state-of-the-art randomized distributed algorithm for monotone submodular function maximization under hereditary constraints, which has an approximation ratio half that of the \greedy algorithm. The \randgreedi\ algorithm randomly partitions the data among all the processors, runs the standard \greedy\ algorithm on each partition independently in parallel, and then executes a {\em single accumulation step} in which all processors send their partial solutions to one processor. 
However, this accumulation step could exceed the memory available on a processor when the memory is small relative to the size of the data, or when solutions are large.
Additionally, the accumulation serializes both the computation and communication and is a bottleneck when scaled to many machines. 

Our  \greedyl\ algorithm brings additional parallelism to this step and can lower the 
memory and running time by introducing hierarchical accumulation organized through an \emph{accumulation tree}. 
Similar to \randgreedi{}, we randomly partition the data among all the 
processors, which constitute the leaves of the accumulation tree. We merge partial solutions at multiple levels in the tree, and the final solution is computed at the root. We prove that the \greedyl\ algorithm has a worst-case expected approximation guarantee of %$\nicefrac{\alpha \cdot b}{(m+b)}$
$(\alpha \: b) / (m+b)$,  % the nicefrac had small fonts
where $\alpha$ is the approximation guarantee for the \textsc{Greedy} algorithm, $b$ is the branching factor, and $m$ is the number of leaves in the accumulation tree.
Using the BSP model, we also analyze the time and communication complexity of the \greedyl\  and \randgreedi\ algorithms and show that the former has lower computation and communication costs than the latter. \looseness=-1

We evaluate the parallel algorithms on three representative and practical submodular function maximization problems: maximum $k$-set cover, maximum $k$-vertex dominating set in graphs, and exemplar-based clustering  (modeled by the $k$-medoid problem).  
We experiment on large data sets with millions of elements that exceed the memory constraints (a few GBs) on a single processor,  
and demonstrate how to choose the accumulation tree to have more levels to adapt to the small memory available on a processor. 
This strategy also enables us to solve for larger values of the parameter  $k$ in the problems discussed above, which corresponds to the size of the solution sought. 
We also show that the number of function evaluations on the critical path of the accumulation tree, and hence the run time,  could be reduced when the parallel algorithm is employed. 
In most cases, we find the quality of the computed solutions by our \greedyl{} closely matches those obtained by the %distributed  
\randgreedi\  algorithm on these problems despite having a worse expected approximation guarantee.

\section{Background and Related Work}
%\subsection{Notations}
\subparagraph*{Submodular functions} A set function $f \colon 2^W \to \mathbb R^+$ 
defined on the power set of a ground set $W$ is \emph{submodular} if it satisfies the \emph{diminishing marginal gain} property.  That is,  $$ f(X \cup \{w\}) - f(X) \geq f(Y \cup \{w\}) - f(Y), \text{ for all } X \subseteq Y \subseteq W \text{ and } w \in W \setminus Y. $$
A submodular function $f$ is \emph{monotone} if for every  $X \subseteq Y \subseteq W$, we have $f(X) \leq f(Y)$. 
The \emph{constrained submodular maximization} problem is defined as follows.
$$
\max f (S) \text{ subject to } S \in \mathcal{C}, \text{where }  \mathcal{C} \subseteq 2^W \text{ is the family of feasible solutions.}
$$  
\noindent We consider {\em hereditary constraints}: 
i.e., for every set $S\in \mathcal{C}$,   every subset of $S$ is also in $\mathcal{C}$. The hereditary family of constraints includes various common ones such as \emph{cardinality constraints} ($\mathcal{C} = \{A \subseteq W: |A| \leq k\}$) and \emph{matroid constraints} ($\mathcal{C}$ corresponds to the collection of independent sets of a matroid).

\subparagraph*{ \lsz extension} For the analysis of our algorithm, we use the {\em \lsz extension}~\cite{Lovasz1983}, a relaxation of submodular functions. A  submodular function $f$ can be viewed as a function defined over the vertices of the unit hypercube,  $f:\{0,1\}^n \rightarrow\mathbb R^+$, by identifying sets $V \subseteq W$ with binary vectors of length $w=|W|$ in which the $i^{\text{th}}$ component is $1$ if $i\in V$, and $0$ otherwise. 
The {\em \lsz extension}~\cite{Lovasz1983}  $\fhat: [0,1]^w  \rightarrow \mathbb{R}^+$ is a convex extension that extends $f$ over the entire hypercube and given by,
$ \fhat(x) = \mathop{\mathbb{E}}\limits_{\theta \in \mathcal{U}[0,1] } \left[ \; f\left(\{i : x_i \geq \theta\}\right) \; \right]. $ %\smf{what is $\cal{U}$}
Here, $\theta$ is uniformly random in $[0,1]$. The \lsz extension $\fhat$ satisfies the following properties~\cite{Lovasz1983}:

\begin{enumerate}[topsep=0pt]\label{eqn:Lova} \item $\fhat(1_S) = f(S)$,   for all $S \subseteq V$ where $1_S\in [0,1]^w$ is a vector containing 1 for the elements in $S$ and $0$ otherwise,  
\item $\fhat(x)\  \text{is convex}$, 
and \item $\fhat(c\cdot x) \geq c \cdot \fhat(x),  \text{for any}\  c\in [0, 1]$. 
\end{enumerate}

An \emph{$\alpha$-approximation} algorithm ($\alpha \in [0,1)$) for  constrained submodular maximization %if it 
produces a feasible solution $S \subseteq W $, satisfying $f(S) \geq \alpha \cdot f(S^*)$, where $S^*$ is an optimal solution. \looseness=-1

\subsection{Related Work}
\label{sec:randgreedi}
\subparagraph*{\greedi{} and \randgreedi{}.} The iterative \greedy\ algorithm for maximizing constrained submodular functions starts with an empty solution. 
Given any current solution $S$, an element is \emph{feasible} if it can be added to the solution without violating the constraints. 
In each iteration, the \greedy\ algorithm on a dataset $V$ chooses a feasible element $e \in V$ that maximizes the \emph{marginal gain}, 
$f(S\cup \{e\}) - f(S)$, w.r.t. the current solution $S$.
The algorithm terminates when the maximum marginal gain is zero or all feasible elements have been considered. \looseness=-1 

\begin{comment}
\begin{algorithm}[ht]
    \caption{ \greedy\ Algorithm }
    \label{Alg:greedy}
    \begin{algorithmic}[1]
  \item[]
    {\small\Procedure{$\greedy\ $}{$V$: Dataset}
    \State $S \leftarrow \emptyset$
    \While{True}
        \State $E \leftarrow \{ e \in V\setminus S: S \cup \{e\} \in \mathcal{C}\}$
        \State $e' \leftarrow \arg\max_{e \in E} f(S \cup \{e\})$ %$ -f(S) $
        \If{ $f(S \cup \{e'\}) = f(S)$ or $E = \emptyset$ }
            \State \textbf{break}
        \EndIf
        \State $S \leftarrow S \cup \{e'\}$
    \EndWhile
    \State \Return{$S$} 
   \EndProcedure}
   \end{algorithmic}
\end{algorithm}
\end{comment}

\begin{comment}
    \begin{algorithm}[t]
\caption{\randgreedi\ framework for maximizing constrained submodular function}
\label{Alg:randgreedi}
\begin{algorithmic}[1]
  \item[]
    { \small \Procedure{\randgreedi}{$V$: Dataset, $m$: number of machines}
    \State $S \leftarrow \emptyset$
    \State Let $\{ P_0, P_1, \ldots, P_{m-1}\}$ be an uniform random partition of $V$.
    \State \textbf{Run} $\greedy(P_i)$ on each machine $i \in [0,m-1]$ to compute the solution $S_i$  
    \State \textbf{Place} $S = \bigcup_i S_i$ on machine 0 
    \State \textbf{Run} $\greedy(S)$ to compute the solution  $T$ on machine 0
    \State \Return{$\arg\max \left\{ f(T), f(S_1), f(S_2), \dotsc, f(S_{m-1}) \right\}$}
\EndProcedure}
  \end{algorithmic}
\end{algorithm}
\end{comment}

We now discuss the \greedi{} and \randgreedi{} algorithms, which are the state-of-the-art distributed algorithms for constrained submodular maximization. The \greedi \ algorithm~\cite{Mirzasoleiman} partitions the data {\em arbitrarily} on available machines, and on each machine, it runs the  \greedy\ algorithm independently to compute a \emph{local} solution.  
These solutions are then accumulated to a single \emph{global} machine. %where they are accumulated.
The \greedy\ algorithm is again executed on the accumulated data to get a \emph{global} solution. The final solution is the best solution among all the local and global solutions.
For a cardinality constraint, where $k$ is the solution size, the \greedi\ algorithm has a worst-case approximation guarantee of $1/\Theta(\min(\sqrt{k}, m))$, where $m$ is the number of machines. \looseness=-1

Although \greedi\ performs well in practice~\cite{Mirzasoleiman}, its approximation ratio is not a constant but depends on $k$ and $m$. 
Improving on this work, 
Barbosa et al.~\cite{barbosa} proposed the \randgreedi \ algorithm, which partitions the data uniformly at random on machines and achieves an \emph{expected} approximation guarantee of $\frac{1}{2} (1- 1/\mathrm e)$ for cardinality and $1/4$ for matroid constraints. In general, it has an approximation ratio of $\alpha/2$ where $\alpha$ is the approximation ratio of the \greedy\ algorithm used at the local and global machines. We present the pseudocode of \randgreedi\ framework in Algorithm~\ref{Alg:randgreedi} in Appendix~\ref{app:pseudocode}. Note that for a cardinality constraint, both \greedi{} and \randgreedi{} perform  $O(nk(k+m))$ calls to the objective function and communicate $O(mk)$ elements to the global machine where $n$ is the number of elements in the ground set, $m$ is the number of machines, and $k$ is solution size.  

Both \greedi{} and \randgreedi{} require a single global accumulation from the solutions generated in local machines that can quickly become dominating since the runtime, memory, and complexity of this global aggregation grows linearly with the number of machines. We propose to alleviate this by introducing a hierarchical aggregation strategy that maintains an accumulation tree. Our \greedyl{} framework generalizes the \randgreedi{} from a single accumulation to a multi-level accumulation. The number of partial solutions to be aggregated depends on the branching factor of the tree, which can be a constant. Thus, the number of accumulation levels grows logarithmically with the number of machines, and the total aggregation is not likely to become a memory, runtime, and communication bottleneck with the increase in the number of machines. We refer to Appendix~\ref{subsec:complexity} for the detailed complexity comparisons of the \randgreedi{} and our \greedyl{} algorithm.

\customhead{Other work} Early approaches on distributed submodular maximization includes the $1/(2+\epsilon)$-approximate \textsc{Sample and Prune} algorithm by  Kumar et al.~\cite{Kumar}, which requires $O(1/\delta)$ rounds assuming  $O(kn^\delta \: \log{n})$ memory per machines.  Here,  $\delta > 0$ is a user parameter. \greedi{}~\cite{Mirzasoleiman} and \randgreedi{}~\cite{barbosa} are shown to be more efficient in practice than the \textsc{Sample and Prune} algorithm. 

More recent distributed approaches~\cite{Chandra, Mokhtari, Alexander} use the multi-linear extension to map the problem into a continuous function. They typically perform a gradient ascent on each local machine and build a consensus solution in each round, which improves the approximation factor to $(1- 1/\mathrm e)$.
However, we do not believe that these approaches are practical since they involve expensive gradient computations   (could be exponential-time). %Even randomized algorithms (by sampling many points) of gradient computations are expensive. 
Most of these algorithms are not implemented, and the one reported implementation solves problems with only a few hundred elements in the data set~\cite{Alexander}.
 
 \section{Description of Our Algorithm}

We describe and analyze our algorithm that generalizes the \randgreedi\ algorithm from a single accumulation step to multiple accumulation steps. Each accumulation step 
corresponds to a level in an \emph{accumulation tree}, which we describe next. We assume that there are $m$ machines identified by the set of ids: $\{0,1,\ldots,m-1\}$.

\customhead{Accumulation tree} An accumulation tree ($T$) is defined by the number of machines ($m$), and branching factor ($b$). It has the same structure as a \emph{complete} $b$-ary tree with $m$ leaves, which means all the leaves are at the same depth.
The tree nodes correspond to processors and the corresponding subset of data accessible to them. The edges of the tree determine the accumulation pattern of the intermediate solutions. The final solution is generated on the root node of $T$. Thus, the branching factor $b$ of the tree indicates the maximum number of nodes that transmit data to its parent. %For each internal node of the tree, we attempt to have exactly $b$ children.  Note that since we plan to construct a complete $b$-ary tree, in the case where $m$ is not multiple of $b$, in each level of the tree, there could be at most one node whose arity is less than $b$. 
The number of accumulation levels (i.e., one minus the height of the tree), denoted by $L$, is $\lceil \log_b{m} \rceil$.

To uniquely identify a node in the tree, we assign an identifier $(\ell, id)$ to each node of $T$, where $\ell$ represents the accumulation level of the node and $id$ represents the machine identifier corresponding to the node. The $id$ for each leaf node is the identifier of the machine to which the leaf node corresponds. All the leaf nodes are at the zeroth level. Each internal node receives the lowest $id$ of its children, i.e., any node $(\ell,i)$ %\smf{Should we use $\ell$?} 
has node $(\ell+1, \lfloor i/b^{\ell+1}\rfloor*b^{\ell+1})$ as the parent. Therefore, the root node will always reside at level $L$ with $id = 0$. Also, we characterize an accumulation tree $T$ by the triple $T(m, L, b)$, where $m$ is the number of leaves (machines), $L$ is the number of levels, and $b$ is the branching factor. 
%%%
\begin{figure}
\centering
\resizebox{0.78\columnwidth}{!}{
\begin{forest}
for tree={rectangle,draw, edge={<-, thick}, minimum size=4.1em, s sep = 1em, inner sep=1pt, l=6em, fill=green!20, font=\small\bfseries}% Rotates the text by 90 degrees}
  [{${2,\!0}$},name=l2, fill=blue!20  [${1,0}$, fill=orange!20     [${0,\!0}$]
      [${0,1}$,name=l0c1]
      [${0,b-1}$,name=l0cb]
    ]
    [${1,b}$, name=l1c1, fill=orange!20     [${0,b}$]
      [${0,b+1}$,name=l0c2]
      [${0,2b-1}$,name=l0c3]
    ]
    [${1,(b\!-\!1)b}$, name=l1cb, fill=orange!20  [${0,b(b-1)}$,name=l0c4]
    [${0,b^2\!-\!b\!+\!1}$,name=l0c5]
      [${0,b^2\!-\!1}$,name=l0c6]]
    ]
  ]
  %\draw(l2) to ([xshift=25ex]l2) node{Level 2};
  \node(first) at ([xshift=65ex]l2) {\textbf{Level 2}};
  \node(second) at ([yshift=-l]first) {\textbf{Level 1}};
  \node(third) at ([yshift=-l]second) {\textbf{Level 0}};
  \draw[thick,loosely dotted] (l1c1) to (l1cb);
  \draw[thick,dotted] (l0c1) to (l0cb);
  \draw[thick,dotted] (l0c2) to (l0c3);
  \draw[thick,dotted] (l0c5) to (l0c6);
\end{forest}

}
    \caption{An accumulation tree with $L=2$ levels, $m=b^2$ machines, and a branching factor $b$. Each node has a label of the form $(\ell, id)$. Here there are $b$ nodes as children at each level, but when there are fewer than $b^L$ leaf nodes, then the number of children at levels closer to the root may be fewer than $b$.}
    \label{fig:recTree}
\end{figure}
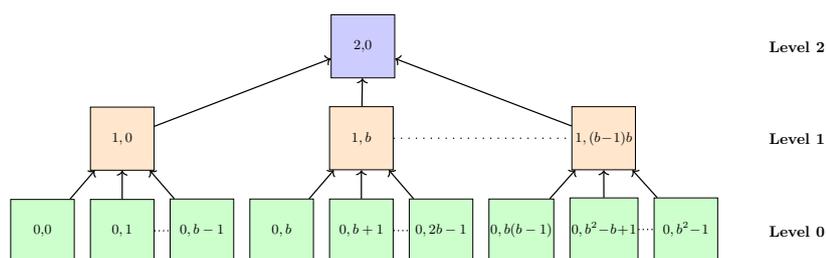

%\todo{Regarding Figure 1. We set up the notation to be $(\ell,id)$. Why are we inconsistently naming nodes as ``$\ell,id$''? SG:Fixed}

% \begin{figure}\footnotesize
% \begin{tikzpicture}[scale=2]
% \foreach \x/\y/\n/\v in { 0/0/0/{0,0} , 1/0/1/{0,1}, 3/0/2/{0,b-1} }
%     \node[draw, ellipse, minimum height=10mm, minimum width=15mm] (\n) at (\x,\y) {$\v $};
% \end{tikzpicture} 
% \caption{Caption}
% \label{fig:my_label}
% \end{figure}

%%%
Figure \ref{fig:recTree} shows an example of a generic accumulation tree with $b^2$ leaves and branching factor $b$. The number of accumulation levels is the level of the root. Here we have $L=\lceil\log_b b^2\rceil = 2$. We also show a few realizations of different accumulation trees in Figure~\ref{fig:8tree} in the Appendix~\ref{app:acc}.

%\color{blue}Figure~\ref{fig:8tree} shows accumulation trees with 8 leaves and with branching factors 2, 3, 4, and 8. The trees with branching factors 2 and 8 have the same branching factor for every internal node as these trees have $b^L$ nodes. But the tree with branching factor 3 has the last node in level 1 with only 2 children. Similarly, the tree with branching factor 4 has the last node in level 2 i.e. the root with 2 children.\color{black}
Observe that the $id$ parameter remains the same in multiple nodes that are involved in computations at multiple levels. %\smf{What does this line mean? SG: It means each node doesn't have a different id. only nodes in each level have different id. We will use the same id in different levels as they identify the machine corresponding to that node.}
For our analysis, we keep the branching factor constant across all levels. 

\begin{comment}
    \smallskip\noindent\textbf{Data Accessibility.} We use $P_{id}$ to denote the elements assigned to the machine $id$. %To indicate the data \emph{accessible} to a particular node in the tree, we describe a set for the input data set as $V_{\ell, id}$. It corresponds to all the data used to compute the solution at the node $(\ell, id)$ and consists of all the elements assigned to its descendants: 
$V_{\ell, id}$ denotes the data \emph{accessible} to a particular node in the tree, which corresponds to all the data used to compute the solution at the node $(\ell, id)$ and consists of all the elements assigned to its descendants:
\begin{align*}
    V_{\ell,id} = \bigcup \limits_{i=0}^{\min(b^\ell -1, m - id)} P_{id+i}.
\end{align*}
\end{comment}

\customhead{Randomness} The randomness in the algorithm is \emph{only} in the initial placement of the data on the machines, and we use a random tape to encapsulate this. 
The random tape $r_W$  has a randomized entry for each element in $W$ to indicate the machine containing that element. 
Any expectation results proved henceforth are over the choice of this random tape. 
Moreover, if the data accessible to a node is $V$, we consider the randomness over just $r_V$. 
Whenever the expectation is over $r_V$, we denote the expectation as $ \mathbb{E}_V$. 

\begin{figure*}
   \resizebox{0.9\linewidth}{!}{ \small $  \greedyl(\ell, id)  = $\\
$\begin{cases}  \greedy(P_{id})  & \ell = 0 \\ 
\arg \max   \begin{cases}
                    \greedy\left( \bigcup \limits_{i \in\{0,1,\dotsc,b-1\}} \greedyl\left(\ell-1,  id + i\cdot b^{\ell-1}\right)\right) \\
                    \greedyl(\ell-1,id)
            \end{cases} & id \text{ mod }b^{\ell} = 0 \\
\text{undefined} & \text{otherwise} \end{cases} $
}
\caption{The recurrence relation for the multilevel \greedyl\ which is defined for each node in the accumulation tree. We denote the random subset assigned to machine $id$ by $P_{id}$.} 
\label{fig:RecRel}
\end{figure*}

\customhead{Recurrence relation}
Figure~\ref{fig:RecRel} shows the recurrence relation that forms the basis of the  \greedyl~ algorithm, defined for every node in the accumulation tree;  it will be the basis for the multilevel distributed algorithm. 
At level $0$ (leaves), the recurrence function returns the \greedy\ solution of the random subset of data $P_{id}$ assigned to it. 
At other levels (internal nodes), it returns the better among the \greedy solution computed from the union of the received solution sets of its children and its solution from its previous level. 
It is undefined for $(\ell, id)$ tuples that do not correspond to nodes in the tree (at higher levels). The detailed pseudocode of our algorithm is presented in Algorithm~\ref{alg:multil2} in Appendix~\ref{app:pseudocode}. \looseness=-1

\begin{comment}
\sg{Do we need this?}We can compare it with the \randgreedi\  algorithm by looking at the recurrence relation at level one. The \randgreedi\ algorithm takes the $\arg\max$ of the accumulated solution and the \emph{best} solution from the children. However, our recurrence relation takes the $\arg\max$ for the accumulated solution and \emph{one} solution from the previous level, thus reducing significant compuation at the internal nodes.
%Our choice reduces the computation at the internal node.
%Our recurrence relation takes the $\arg\max$ for the accumulated solution and \emph{one} solution from the previous level.
However, 
For trees with only two levels, We show that this modification preserves the approximation guarantee of the \randgreedi\ algorithm.
\end{comment}

\section{Analysis of Our Algorithm} 
%In this section, we will derive the expected approximation ratio of the \greedyl{} algorithm. We will then describe the three submodular functions we experiment with and derive their computation and communication complexities.

%\subsection{Expected Approximation Ratio} 

We prove the \emph{expected approximation ratio} of \greedyl\ algorithm in Theorem~\ref{thm:approx} using three Lemmas.
We restate a result from~\cite{barbosa} that applies to the leaves of our accumulation tree and characterizes elements that do not change the solution computed by the \greedy algorithm. \looseness=-1
\begin{lemma}[\cite{barbosa}]
\label{lem:stable}
If we have \greedy$(V \cup \{e\}) = \greedy(V) $, for each element $e\in B$, then $\greedy(V \cup B) = \greedy(V)$.
\end{lemma} 

\begin{comment}
\begin{proof}
 \color{red} We borrow this proof from the Barabosa\cite{barbosa} paper. Suppose that $\greedy(V \cup B) \neq \greedy(V)$. Let $e$ be the first element of $B$ selected by $\greedy(V \cup B)$. Then
$\greedy(V \cup \{e\})$ will also select $e$, since only elements of $V$ were selected before $e$ was selected. Therfore, this  contradicts the fact that $\greedy(V \cup \{e\}) = \greedy(V).$
\end{proof}
\end{comment}

The next two Lemmas connect the quality of the computed solutions to the optimal solution at the internal nodes (in level one) of the accumulation tree. 
Lemma~\ref{lem:individual} provides a lower bound on the expected function value of the \emph{individual} solutions of the  \greedy algorithm received from the leaf nodes while  
Lemma~\ref{lem:accumulate} %provides a lower bound on the expected function value of the solution set from the \greedy\ algorithm executed at an internal node at level one on the union of the partial solutions from its children. 
analyzes the expected function value of the \greedy{} execution over the \emph{accumulated} partial solutions. Due to space constraints, we state the lemmas here and provide the proofs in the Appendix~\ref{app:omit-proofs}.

Let $p\colon V  \to [0,1]$ be a probability distribution over the elements in $V$, and  $ A \sim V(1/m)$ be a random subset of $V$ such that each element is independently present in $A$ with probability $1/m$. %This probability is defined for all elements in the optimal solution(OPT) of the problem.
The probability $p$ is defined as follows:  
\begin{equation*}
    p(e) = 
    \begin{cases} 
    \Pr\limits_{A \sim V(1/m)}
 \left[e \in \greedy(A\cup \{e\}) 
\right],  & \text{if} \  e\in OPT;\\ 
 0,  & \text{otherwise.} 
\end{cases} 
\end{equation*}
For any leaf node, the distribution $p$ defines the probability that each element of $OPT$ is in the solution of the \greedy\ algorithm when it is placed in the node. 

%Next, we state Lemma~\ref{lem:individual} that relates the expected solution of the \greedy\ algorithm at a leaf node with the optimal solution when the approximation ratio of the \greedy\ algorithm at the child is $\beta$. The proof follows from the same lemma as stated in Barabosa et al.\cite{barbosa}.
\begin{lemma}
\label{lem:individual}
Let $c$ be a leaf node of the accumulation tree, $S_c$ be the solution computed from $c$, and $V_c \subset V_{n}$ be the elements considered in forming $S_c$. If \greedy\ is an $\alpha$-approximate algorithm, then 
$ {\mathbb{E}}_{V_n}[f(S_c)] \geq \alpha \cdot \fhat(1_{OPT}-p). 
$
\end{lemma}

%Now we show how the solution of the \greedy{} algorithm runs at each internal node at level 1 and compares with the optimal solution.
\begin{lemma}
\label{lem:accumulate}
 Let $D$ be the union of all the solutions computed by the $b$ children of an internal node $(1,id)$ in the accumulation tree, and $S$ be the solution from the \greedy{} algorithm on the set $D$. If \greedy\ is an $\alpha$-approximate algorithm, then $\mathbb{E}_{V_n}[f(S)] \geq \dfrac{b*\alpha}{m} \cdot \fhat(p).$ 

 %Proved in Appendix Section \ref{lem:accumulate:proof}.
\end{lemma}

\begin{theorem}
\label{thm:approx}
Let $T(m, L,b)$ be an accumulation tree,   
 $V$ be the ground set, and $r_V$ be a random % tap
mapping of elements of $V$ to the leaves of the tree $T$.
Let $OPT$ be an optimal solution computed from $V$ for the constrained submodular function $f$.  
If \greedy\ is an $\alpha$-approximate algorithm,  then 
$
{ \mathbb{E}\left[f(\greedyl(V))\right] \geq \dfrac{b\cdot \alpha}{(m+b)} f(OPT).}
$
\end{theorem}

\begin{proof}
\begin{comment}
From Lemma \ref{lem:individual} we get,
\begin{align}
    \mathbb{E}[f(S_c)] \geq \alpha\cdot \fhat(1_{OPT}- p). \label{eqn:EfSc}
\end{align}  
\end{comment}

We concentrate on a node at level 1, where after obtaining the partial solutions from the children of this node, we compute the \greedy{}   on the union of these partial solution. Let $S_c$ be any of the partial solutions, $S$ be the union of these partial solutions, and $T$ be $\arg \max \{f(S), f(S_{c})\}$. From Lemma~\ref{lem:individual} and Lemma~\ref{lem:accumulate},
\begin{comment}
From Lemma~\ref{lem:accumulate}, we have
\begin{equation}
    \mathbb{E}[f(S)] \geq \frac{b\alpha}{m}\cdot \fhat(p). \label{eqn:EfS}
\end{equation} 
The solution of level one is $T = \arg \max \{f(S), f(S_{c})\}$. Then, we can use the lower bounds calculated earlier in Eqn. \ref{eqn:EfSc} and Eqn. \ref{eqn:EfS} to find lower bounds for $T$. 
\end{comment}
\[
    \mathbb{E}[f(T)] \geq \frac{b\alpha}{m} \cdot \fhat(p)  \text{ and } \mathbb{E}[f(T)] 
    \geq \alpha\cdot \fhat(1_{OPT} - p). 
\]
By multiplying the first inequality  by $m/b$ and then adding it to the second, we get
\begin{align*}
   (m/b + 1 )\mathbb{E}[f(T)] &\geq \alpha\cdot ( \fhat(1_{OPT} - p) + \fhat(p)) 
     = \alpha\cdot \fhat(1_{OPT}) 
    &\pushright{[\text{\lsz Ext.~(2)}, {\ref{eqn:Lova}}]}\\ 
    \mathbb{E}[f(T)] & \geq \frac{\alpha}{(m/b+1)} \cdot \fhat(1_{OPT}).
\end{align*}
%Dividing by  by $m/b+1$,  and substituting from \lsz Ext.~(1), \ref{eqn:Lova} we get T to be is $\alpha/ (m/b+1) $-approximate.

The theorem follows since the solution quality can only improve at higher levels of the tree. \looseness=-1

\end{proof}

\section{Experimentation}
%\subsection{Experimental setup}
\label{sec:setup}
\subparagraph*{Experimental Setup.} We conduct experiments to evaluate our algorithms using different accumulation tree structures and compare them with \greedy{} and \randgreedi{} to assess the quality, runtime, and memory footprints of these algorithms.  
All the algorithms are executed on a cluster computer %~\cite{McCartney2014} 
%of Purdue University 
with $448$ nodes, each of which is an AMD EPYC $7662$ node with $256$ GB of total memory shared by the $128$ cores. Each core operates at $2.0$ GHz frequency.  %The cores on a node are organized hierarchically: four cores form a core complex, two core complexes form a core complex die, eight core complex dies form a socket, and two sockets constitute a node. Unfortunately, there are only $16$ memory controllers for the $128$ cores, and hence in this NUMA architecture, memory contention is an issue on cores within a node. 
To simulate a completely distributed environment on this cluster, we needed to ensure that the memory is not shared between nodes. Therefore, in what follows, a machine will denote one node with just one core assigned for computation, but having access to all $256$ GB of memory. We also found that this made the run time results more reproducible. 

For our experimental evaluation, we report the \emph{runtime} and \emph{quality} of the algorithms being compared. For runtime, we exclude the file reading time in each machine, and for the quality, we show the objective function value of the corresponding submodular function. Since the \randgreedi{} and \greedyl{} are distributed algorithms, we also report the \emph{number of function calls in the critical path} of the computational tree, which represents the parallel runtime of the algorithm. Given an accumulation tree, the number of function calls in the critical path refers to the maximum number of function calls that the algorithm makes along a path from the leaf to the root.  
In our implementation, this quantity can be captured by the number of function calls made by the nodes of the accumulation tree with $id=0$ since this node participates in the function calls from all levels of the tree. 

\begin{table}[t]
    \centering
    \resizebox{0.75\linewidth}{!}{
    \begin{tabular}{ll|rrr}
         \toprule
         Function & Dataset & $n= |V|$ & $\sum_u \delta(u)$ & avg. $\delta(u)$ \\%& Size  \\ 
         \midrule
         \multirow{7}{22mm}{$k$-dominating set} &  AGATHA\_2015 & 183,964,077 & 11,588,725,964 & 63.32 \\
         & MOLIERE\_2016 & 30,239,687 & 6,669,254,694 & 220.54 \\
         %& GAP-urand  &134,217,728 & 4,294,966,740 & 31.90 \\
         & com-Friendster & 65,608,366 & 1,806,067,135  & 27.52 \\%& 30G \\ 
         &road\_usa & 23,947,347 & 57,708,624 & 2.41 \\%& 334M\\
         &road\_central & 14,081,816 & 33,866,826 & 2.41 \\%& 144M\\
         &belgium\_osm & 1,441,295 & 3,099,940 &  2.14 \\%& 17M\\
         \midrule
         \multirow{3}{22mm}{$k$-cover  } 
         &webdocs & 1,692,082 & 299,887,139 & 177.22 \\%& 2G\\
         &kosarak & 990,002 & 8,018,988 & 8.09 \\%& 49M\\
         &retail & 88,162 & 908,576 & 10.31 \\%& 7M\\ 
         \midrule
         $k$-medoid & Tiny ImageNet & 100,000 & 1,228,800,000 & 12,288\\%& 32G\\
         \bottomrule   
    \end{tabular}
    }
    \caption{Properties of Datasets used in the experiments. $\delta(u)$ is the number of neighbors of vertex $u$ for the $k$-dominating set problem, the cardinality of the subset $u$ for the $k$-cover problem, and the size of the vector representation of the pixels of image $u$ for the $k$-medoid problem. %G stands for Gigabytes, and M for Megabytes.
    }
    \label{tab:dataDetails}
    
\end{table}

\customhead{Datasets} In this paper, we limit our experiments to cardinality constraints using three different submodular functions described in detail in Appendix~\ref{subsec:complexity}. 

Our benchmark dataset is shown in Table~\ref{tab:dataDetails}. 
%\sg{This is ordered by $\sum\delta(u)$ not $n$. Let me know if it needs to change.}
They are grouped based on the objective function and are sorted by the  $\sum_{u} \delta(u)$ values within each group (see the Table for a definition). 
%size of the dataset in each group. 
For the $k$-dominating set, our testbed consists of the Friendster social network graph~\cite{friendster}, a collection of road networks from \emph{DIMACS10} dataset and the \emph{Sybrandt} dataset. 
We chose road graphs since they have relatively small average vertex degrees, leading to large vertex-dominating sets. We chose the \emph{Sybrandt} collection \cite{Agatha}\cite{Moliere} since it is a huge data set of machine learning graphs. %\smf{Please provide reference. Largest is very strong word.}
For the $k$-cover objective, we use popular set cover datasets from the \emph{Frequent Itemset Mining Dataset Repository} \cite{FIMI}. % which contains popular benchmarks for set covers.  
%We choose \emph{webdocs}\cite{webdocs}, \emph{retail} \cite{retail}, and \emph{kosarak}. 
 For the $k$-medoid problem, we use the \emph{Tiny ImageNet} dataset \cite{tiny}. \looseness=-1%, which contains $100,000$ images with $200$ different classes and $500$ images from each class. Each image is $64 \times 64$ pixels in size. 

\customhead{MPI Implementation} \greedyl{} is implemented using C\texttt{++11}, and compiled with g\texttt{++9.3.0}, using the {\tt O3} optimization flag. We use the Lazy Greedy~\cite{lazygreedy} variant that has the same approximation guarantee as the \greedy{} but is faster in practice since it potentially reduces the number of function evaluations needed to choose the next element (by using the monotone decreasing gain property of submodular functions). Our implementation of the \greedyl{} algorithm uses Open MPI implementation for the inter-node communication. We use the MPI\_Gather and MPI\_Gatherv primitives to receive all the solution sets from the children (Line~\ref{lin:rec} in Algorithm~\ref{alg:multil2}). We generated custom MPI\_Comm communicators to enable this communication using MPI\_Group primitives. Customized communicators are required since every machine has different children at each level. Additionally, we use the MPI\_Barrier primitive to synchronize all the computations at each level. \looseness=-1

\subsection{Experimental Results}
%\sg{Do we need setup and results in separate sections?}
The experiments are executed with different accumulation trees that vary in the number of machines ($m$),  the number of levels $(L)$, and the branching factors $(b)$, in order to assess their performance. We repeat each experiment six times and report the geometric mean of the results. Unless otherwise stated, a machine in our experiments represents a node in the cluster with only one core assigned for computation.% as stated in Section~\ref{sec:setup}.  
Whenever memory constraints allow,  we compare our results with the sequential \greedy\  algorithm that achieves $(1-1/e)$ approximation guarantee. \looseness=-1  

Recall that our \greedyl{} algorithm generalizes the \randgreedi{} algorithm by allowing multiple levels in the accumulation tree, thus removing the bottleneck of a single aggregation. In the following, we verify this through a series of experiments. 

In Section~\ref{sec:NoRes}, we assess the performance of our algorithm using different accumulation tree structures. We fix the number of machines and construct the best parameters of the accumulation tree for our dataset. Additionally, the experiment also demonstrates that the number of function calls in the critical path is a good estimate of the parallel runtime.
In Section~\ref{sec:mem-exp}, we show the memory benefit of our \greedyl{} w.r.t \randgreedi{} with two experiments.
In Section~\ref{sec:vary-k}, we impose a limit of 100 MB space for each node and vary $k$, the selection size. This also simulates how the new algorithm can find applications in the \emph{edge-computing} context. In Section~\ref{sec:mem-lim}, we fix $k$  and vary the memory limits, necessitating different numbers of nodes to fit the data in the leaves. We observe the quality and runtime of different accumulation tree structures in these two experiments. Both these experiments are designed to show that the \randgreedi{} algorithm quickly runs out of memory with increasing $m$ and $k$, and by choosing an appropriate accumulation tree, our \greedyl\ algorithm can solve this problem with negligible drop in accuracy. For these experiments, we will choose the computational tree with the lowest depth that can be used with the memory limit and $k$ values. \looseness=-1

In Section~\ref{sec:scaling}, we perform a scaling experiment by varying the number of machines and using the tallest tree by setting a branching factor of two for the accumulation tree. We specifically show that even though the \randgreedi{} algorithm has a low asymptotic communication cost, it can become a bottleneck when scaled to a large number of machines. We also show how our algorithm alleviates this bottleneck. Finally, in Section~\ref{sec:k-medoid}, we perform experiments for the k-medoid objective function and show that we can provide a significant speedup by using taller accumulation trees without loss in quality. The k-medoid function is extensively used in machine learning as a solution to exemplar-based clustering problems. \looseness=-1

\subsubsection{Accumulation tree parameter selection}
\label{sec:NoRes}
\begin{figure}[t]
    \centering
    % \begin{subfigure}[t]{0.50\linewidth}
    %     \includegraphics[width=\textwidth]{images/SummCalls.png}
    % \end{subfigure}
    % \begin{subfigure}[t]{0.49\linewidth}
    %     \includegraphics[width=\textwidth]{images/SummTime.png}
    % \end{subfigure}
    \includegraphics[width=\textwidth]{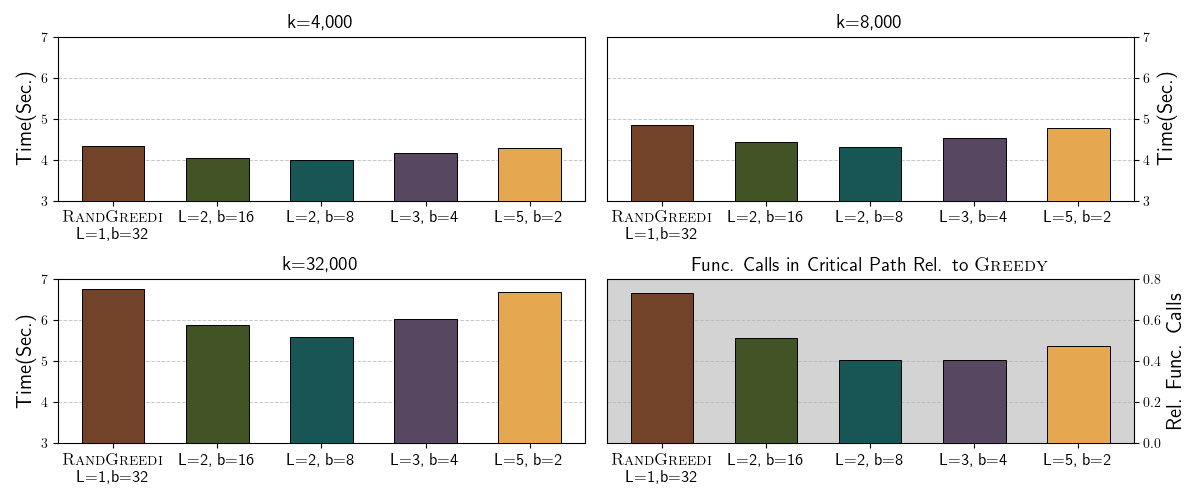}
    \caption{Geometric means of results from \greedyl\ for $k$-dominating set (on road datasets) and for $k$-cover (on set cover benchmarks) using $32$ machines. The $L$ and $b$ represent different accumulation tree configurations. The bottom right plot shows the geometric means of the number of function calls in the critical path relative to the \greedy{} algorithm for $k=32,000$.The remaining three subfigures show the execution times for different $k$ values and accumulation trees.  }
    \label{fig:Summary}
\end{figure}

In this experiment, we show results for the $k$-dominating set and $k$-coverage problem by fixing the number of machines and varying branching factors, the number of levels in the accumulation tree, and the selection set size $k$.
%We obtain results for the six datasets for $k$-dominating set and $k$-coverage detailed in Table~\ref{tab:dataDetails}. 
In Figure \ref{fig:Summary}, we provide summary results on the number of function evaluations in the critical path relative to the \greedy\ algorithm and the running times by taking a geometric mean over all nine datasets. \looseness=-1

Three subfigures (top left, top right, and bottom left) of Figure~\ref{fig:Summary} show the \emph{execution time} in seconds for the \greedyl\ and \randgreedi\ algorithms, as the number of levels and the parameter $k$ are varied. When $k$ is small (top left), there is less variation in the execution time
since work performed on the leaves dominates overall time. As $k$ increases (bottom left), the \greedyl\ algorithm becomes faster than the \randgreedi\ algorithm ($L= 1, b=32$). 
Note that although Figure~\ref{fig:Summary} presents the geometric mean results over all nine datasets, the runtime and the function values for the individual datasets follow the same trend. The largest and smallest reduction in runtime we observe is on the belgium\_osm and kosarak datasets with a reduction of around $22\%$ and 1\% across, respectively, for all $k$ values. \looseness=-1

The bottom right plot fixes $k=32,000$ and shows the \emph{number of function calls} in the critical path of the accumulation tree relative to the \greedy algorithm for different $(L,b)$ pairs. Here, the leftmost bar represents the \randgreedi{} algorithm. We observe that the relative number of function calls for \randgreedi{} is around 70\% of \greedy{}, whereas the \greedyl{} (with $L=2$ and $b=8$) reduces it by 15 percent. From Table~\ref{tab:Complexity},  the function calls complexity at a leaf node is $O(nk/m)$ and at an accumulation node, it  is $O(mk^2)$ for the \randgreedi{} algorithm. Hence, the accumulation node dominates the computation since it has a quadratic dependence on $k$, becoming a bottleneck for large $k$ values. 
This plot also shows that the number of function calls is a good indicator of the algorithm's run time and that the cost of function evaluations dominates the overall time. The other factor affecting run time is communication costs, which are relatively small and they grow with the number of levels when $k$ is very large.  %\smf{I don't see any communication cost in Figure~\ref{fig:Summary}. SG: The number of calls is a good indicator for time but is not exact since time also includes communication costs. The sentence just mentions how that communication cost changes with number of levels}

We note (not shown in the figure) that the objective function values obtained by the \greedyl \ algorithm are not sensitive to the choice of the number of levels and the branching factors of the accumulation tree and differ by less than $1\%$ from the values of the \randgreedi\ algorithm. For the {webdocs} $k$-coverage problem, however,  \greedy\ quality is about $20\%$ higher than both the \randgreedi\ and \greedyl{}.

\begin{comment}
road_usa
         4 &  & 5175 \\
         8 & 4878 & 5175 \\
         16 & 5081 & 5175 \\
         32 & 5175 & 5175 \\ \hline 
\end{comment}

In Table~\ref{tab:quality-levels} in Appendix~\ref{app:omit-tab}, we report the advantages of accumulating in multiple levels over choosing to stop at level one of accumulation. We use maximum $k$-cover function % \smf{What is the function; k-cover?} 
for the Friendster dataset with $k=1000$ for different branching factors at the first and final accumulation levels. We observe the objective values at the highest accumulation level are not very sensitive to the tree parameters, contrary to their sensitivity  to the approximation ratio derived in Theorem~\ref{thm:approx}.  %\smf{What do you mean? you should say that the objective values are increasing as you increase the number of levels. Right? SG: I want to state that the practical results don't show a change with tree parameters while appoximation ratio does change.}  

\begin{figure}[t]
    \centering
    % \begin{subfigure}[t]{0.6\textwidth}
    %     \includegraphics[width=0.95\textwidth]{images/roadCalls3.png}
    % \end{subfigure}
    % \begin{subfigure}[t]{0.39\textwidth}
    %     \includegraphics[width=0.95\textwidth]{images/roadVals.png}
    % \end{subfigure}
    \includegraphics[width=\textwidth]{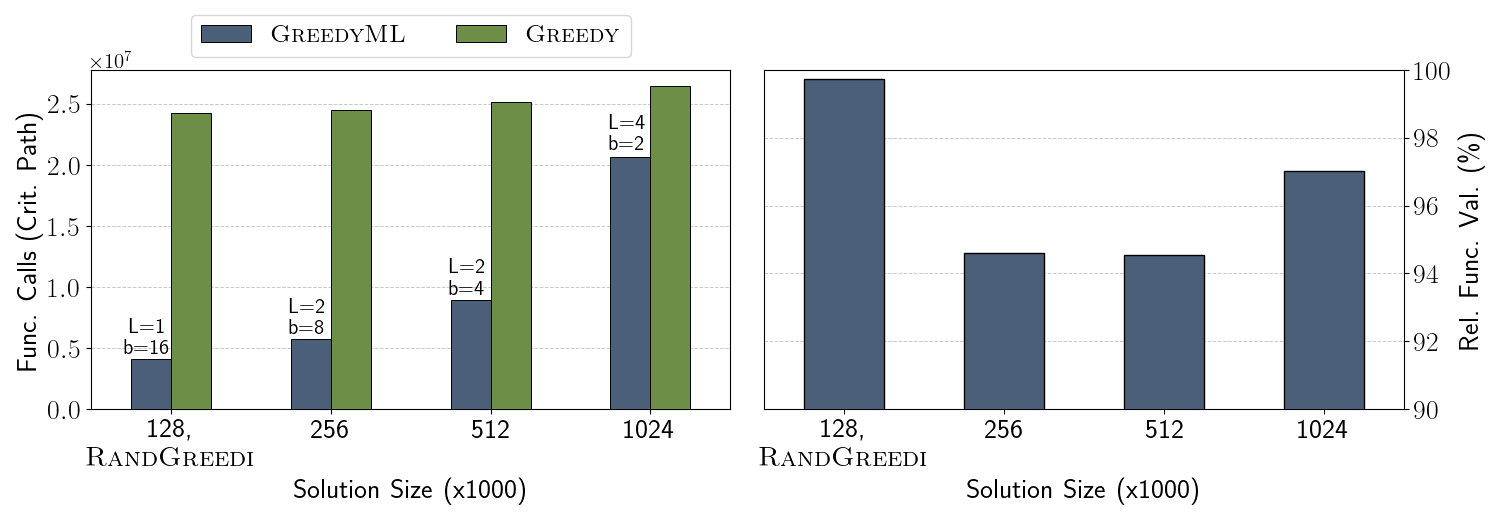}
    \caption{Results from \greedyl{} for the $k$-dominating set problem on the road\_usa dataset on 16 nodes with varying $k$. The pair($L$,$b$) shows the number of levels and branching factors chosen for specific $k$ values. The function values are relative to the \greedy{} algorithm. Note that the leftmost bars in both plots represent the \randgreedi{} results.}
    \label{fig:road}
\end{figure}

\subsubsection{Experiments with memory limit}
\label{sec:mem-exp}

\subparagraph*{Varying \emph{k}.}\label{sec:vary-k} For this experiment, we use 16 machines with a limit on the available memory of 100 MB per machine and vary $k$ from $128,000$ to $1,024,000$ for the $k$-dominating set problem on the road\_usa~\cite{roadusa} dataset. %Note that the $k$-values other than $128,000$ \smf{bigger than?} cause the \randgreedi{} algorithm to run out of memory. 
The small memory limit in this experiment can also be motivated from an \emph{edge computing} context. \looseness=-1

The left plot in Figure~\ref{fig:road} shows the number of function calls with varying values of $k$ for the \greedy{} (green bars) and \greedyl{} algorithms (blue bars). For the \greedyl{} (and the \randgreedi{}), we are interested in the number of function calls in the critical path since it represents the parallel runtime of the algorithm. With our memory limits, only $k=128,000$ instance can be solved using the \randgreedi{} algorithm. \looseness=-1

As we increase $k$, we are able to generate solutions using our \greedyl{} with different accumulation trees. The corresponding lowest-depth accumulation tree with the number of levels and branching factor is shown on top of the blue bars. The result shows that the number of function evaluations on the critical path in the \greedyl\ algorithm is smaller than the number of function evaluations in the sequential \greedy\ algorithm. While the number of function calls for accumulation trees with smaller $b$ values is larger than \randgreedi{}, we see that \greedyl{} can solve the problems with larger $k$ values in the same machine setup, which was not possible with \randgreedi{}. But it comes with a trade-off on parallel runtime. We observe that as we make the branching factor smaller, the number of function calls in the critical path increases, suggesting that it is sufficient to choose the accumulation trees with the largest branching factor (thus the lowest depth tree) whenever the memory allows it. \looseness=-1

The right plot of Figure~\ref{fig:road} shows the relative objective function value, i.e., the relative number of vertices covered by the dominating set compared to the \greedy{} algorithm, with varying $k$. The figure shows that the \randgreedi\ and \greedyl\ algorithms attain quality at most $6\%$ less than the serial \greedy\ algorithm. Similar trends can be observed for other datasets.\looseness=-1 %in the summary of results shown in Figure \ref{fig:Summary}. 

\begin{table}[t]
        \centering
        \resizebox{.8\textwidth}{!}{
        \begin{tabular}{p{24mm}|p{10mm}p{20mm}>{\raggedleft}p{4mm}p{2mm}p{3mm}>{\raggedleft}p{22mm}>{\raggedleft}p{14mm}}
        
        \toprule
            Dataset &Alg. &Mem. Limit  & $m$ & $b$&  $L$ & Rel. Func.(\%)& Time (s.)  \tabularnewline
            \midrule
            \multirow{3}{19mm}{Friendster} 
            &RG &4GB  	&8  &8 &1 &99.959  &61.994\tabularnewline
            &GML &2GB  	&16 &4 &2 &99.903  &61.352
  \tabularnewline
            &GML &1GB  	&32 &2 &5 &99.793  &79.997
 \tabularnewline
            \midrule
            \multirow{3}{19mm}{MOLIERE\_2016}
            &RG & 8GB	&8  &8 &1 &99.257  & 121.318
 \tabularnewline
            &GML &  4GB	&16 &4 &2 &99.106  & 108.764
 \tabularnewline
            &GML &	2GB &32 &2 &5 &98.990  & 161.139
 \tabularnewline
         \midrule
        \multirow{3}{19mm}{AGATHA\_2015} 
            &RG &  12GB	&8  &8 &1 &99.996  &94.122 \tabularnewline
            &GML &  6GB	&16 &4 &2 &99.995  &99.574 \tabularnewline
            &GML  & 3GB &32 &2 &5 &99.989  &104.156 \tabularnewline
            \bottomrule

        \end{tabular}
        }
        \caption{Results for $k$-dominating set on the Friendster, road\_usa and webdocs datasets. The memory size per machine is varied for the Friendster dataset. The number of machines $m$ and the accumulation tree are selected based on the size of the data and the size of the solutions to get three different machine configurations. We report the function values relative to the \greedy\ algorithm and the execution time in seconds. Note that the 4GB entries run with $L=1$ and correspond to the \randgreedi{} (RG) algorithm. We use the same three machine organizations for the road\_usa and webdocs datasets to show they follow similar trends in solution quality and execution time.}
        \label{tab:expt1}
\end{table}
\phantomsection
\subparagraph*{Varying memory limits.}
\label{sec:mem-lim}
This experiment demonstrates that the memory efficiency of the \greedyl\ algorithm enables us to solve problems on parallel machines, whereas the \randgreedi{}  and \greedy{} cannot solve them due to insufficient memory. Unlike the previous experiment (Varying $k$), where we selected the accumulation trees based on $k$, here, we fix $k$  and choose accumulation trees based on the memory available on the machines. We consider the $k$-dominating set problem and %first 
report results on the Friendster~\cite{friendster}, AGATHA\_2015\cite{Agatha}, and MOLIERE\_2016\cite{Moliere} dataset in Table~\ref{tab:expt1}. For the Friendster dataset, we choose $k$ such that the $k$-dominating set requires $512$ MB, roughly a factor of $64$ smaller than the original graph. 
The \randgreedi{}  algorithm (the first row) can execute this problem only on $8$ machines, each with $4$ GB of memory, since in the accumulation step, one machine receives solutions of size $512$ MB each from $8$ machines. The \greedyl{}  algorithm having multiple levels of accumulation can run on $16$ machines with only $2$ GB memory, using $L= 2$ and $b=4$. Furthermore, it can also run on $32$ machines with only $1$ GB memory, using $L=5$ and $b=2$. We repeat the same experiment for the other two datasets with these three machine configurations with corresponding memory restrictions. 

We show relative quality and running time for the three datasets from these configurations in  Table~\ref{tab:expt1}. 
Our results show that function values computed by the \greedyl{}  algorithm (the $2$ and $1$ GB results) are insensitive to the number of levels in the tree. As expected, increasing the number of levels in the accumulation tree increases the execution times due to the communication and synchronization costs involved. However, aggregating in multiple levels enables us to solve large problems by overcoming memory constraints. So, in this scenario, it is sufficient to select the number of machines depending on the size of the dataset and then select the branching factor such that the accumulation step does not exceed the memory limits.
We also notice that the \randgreedi\ algorithm has an inherently serial accumulation step, and the \greedyl\ algorithm provides a mechanism to parallelize it. \looseness=-1 

\begin{figure}[tb]
    \centering
    \includegraphics[width=\textwidth]{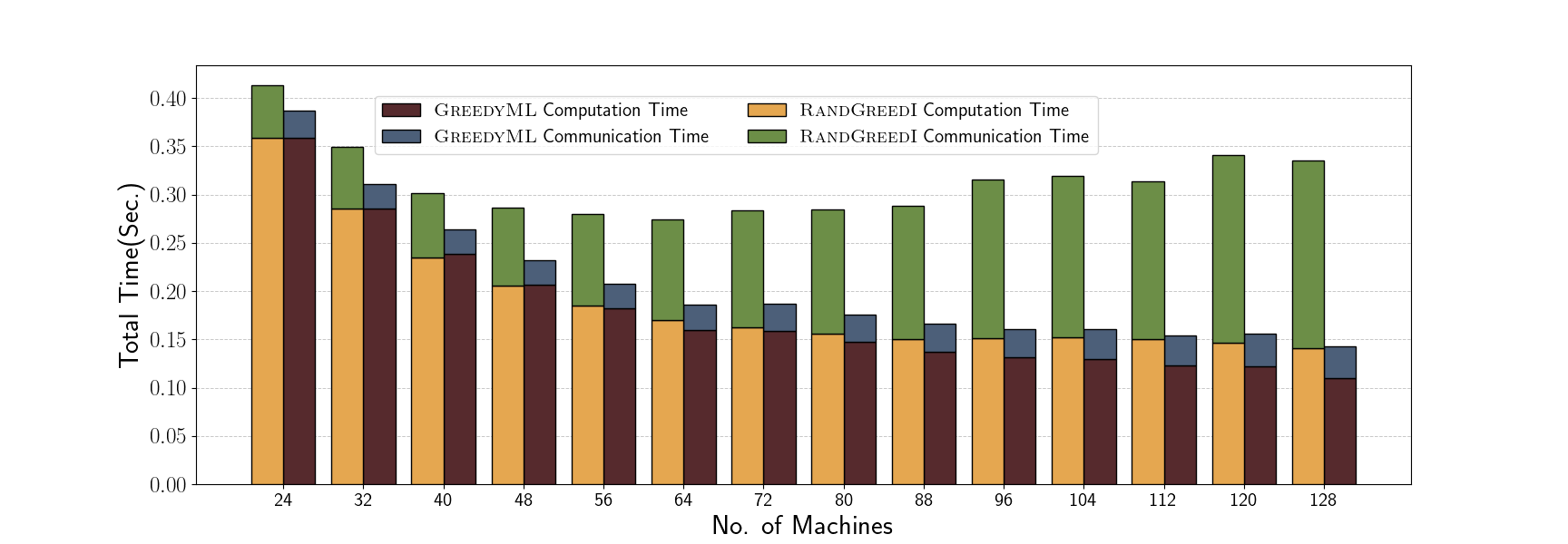}
   
    \caption{Strong scaling results of the \randgreedi{} and \greedyl{} algorithms for $k=50$ on Friendster dataset for the $k$-dominating set problem. We set $b=2$ for the \greedyl{} algorithm. }
    \label{fig:Scaling}
\end{figure} 

\subsubsection{Scaling results} 
\label{sec:scaling}
Next, we show how the \greedyl{} algorithm alleviates the scaling bottlenecks of the \randgreedi{} algorithm using the $k$-dominating set problem on the Friendster dataset. We set the branching factor $b=2$ for the \greedyl{} algorithm since this has the highest number of levels and, thus, the lowest approximation guarantee. We compare communication and computation times against \randgreedi{} algorithm from 8 to 128 machines with $k=50$. %We compare the total execution time, communication time, and computation time for the \greedyl{} and the \randgreedi{} algorithms.

In Figure~\ref{fig:Scaling}, we plot the total execution time by stacking communication and computation times for the two algorithms. For \randgreedi{}, the communication time scales poorly since it increases linearly with the number of machines (See Table~\ref{tab:Complexity}). But, for \greedyl{} algorithm (with a constant branching factor, $b=2, L=\log_2 m$), the communication cost is $O(k\log m)$, which grows logarithmically in the number of machines. Figure~\ref{fig:Scaling} shows that the total communication times of the \greedyl{} algorithm are consistently around $0.25$ seconds, whereas the \randgreedi{} increases from $0.05$ second to $2$ seconds. We observe that computation times for both \randgreedi{} and \greedyl{} changes similarly with $m$, indicating that the majority of the computation work is performed at the leaf nodes. For computation time, we observe a slightly worse scaling of \randgreedi{} compared to \greedyl{}, again because the central node becomes a computational bottleneck as $m$ increases. Similar to other experiments, we observe (not shown in the plot) an almost identical quality in the solutions,  where the \greedyl{} solution has a quality reduced by less than $1\%$ from that of the \randgreedi{} algorithm.

\subsubsection{The \textit{k}-medoid problem}
\label{sec:k-medoid}

\begin{table}[t]
    \centering
    \resizebox{0.6\columnwidth}{!}{
    \begin{tabular}{p{3mm}p{4mm}|p{16mm}p{14mm}|p{16mm}p{14mm}}
    \toprule
       $L$ & $b$ & \multicolumn{2}{c|}{Local Obj.} & \multicolumn{2}{c}{Added Images}\\
       
       & & Rel. Func. Val. (\%) & Speedup & Rel. Func. Val. (\%)& Speedup \\
       \midrule
        5	&2	&92.22	&2.00 &93.69		&2.01 \\
        3	&4	&92.21	&1.96 &92.70		&1.94 \\
        2	&8	&92.73	&1.95 &92.77		&1.93 \\
        2	&16	&92.22	&1.49 &93.34		&1.44 \\
        \bottomrule

    \end{tabular}
    }
    \caption{Results from \greedyl\ for the $k$-medoid function on the \emph{Tiny ImageNet} data set using different accumulation trees. 
    The table shows the relative function values and speedup compared to the \randgreedi\ algorithm using two different schemes for computing the local objective functions on 32 nodes. Higher values are better for both schemes. Recall that $L$ and $b$ are the number of levels and branching factor, respectively.}
    \label{tab:tinyImg}
\end{table}

%\SG{Images in 2x8 or 4x4 grid?}

Our final experiment considers the $k$-medoid function that solves the exemplar-based clustering problem. 
Our dataset consists of the \emph{Tiny ImageNet} dataset \cite{tiny} containing 100K images ($64 \times 64$ pixels) with $200$ different classes with $500$ images from each class. %Each image is $64 \times 64$ pixels. 
We convert and normalize the image into a vector and use Euclidean distance to measure dissimilarity. We define an auxiliary vector $e_0$ as a pixel vector of all zeros. Note that, unlike the other two functions, the $k$-medoid function requires access to the full dataset to compute the functional value. Since the dataset is distributed, this poses an issue in the experiment. To overcome this, following~\cite{Mirzasoleiman,barbosa}, we calculate the objective function value using only the images available \emph{locally} on each machine. This means the ground set for each machine is just the images present in that machine.  %This is motivated by an analysis from Mirzasoleiman et al. (Theorem 10, \cite{Mirzasoleiman}) showing that computing $f(S)$ with the ground set as some subset $D \subseteq V$ chosen uniformly at random provides a high-probability additive approximation to the function value $f(S)$ evaluated with ground set $V$. 
Additionally, they~\cite{Mirzasoleiman,barbosa} have also added subsets of randomly chosen images to the central machine to provide practical quality improvement. We have followed these techniques (\emph{local only} and \emph{local with additional images}) in the experiments for our multilevel \greedyl\ algorithm. 

In our experiments, we set $k$ to $200$ images, fix the number of machines ($m = 32$), and vary the accumulation trees by choosing different $L$ and $b$. % We set the solution size $k$ to $200$ images. 
For the variant with additional images, we add $1,000$ random images from the original dataset to each accumulation step. \looseness=-1

In Table \ref{tab:tinyImg}, we show the relative objective function values and speedup for different accumulation trees relative to the \randgreedi{} algorithm. We observe that the objective function values for \greedyl{} algorithm are almost similar to \randgreedi{}. Our results show that the \greedyl\ algorithm becomes gradually faster as we increase the number of levels with runtime improvement ranging from $1.45-2.01\times$.  This is because the $k$-medoid function is compute-intensive, where computation cost increases quadratically with the number of images (Table~\ref{tab:Complexity}). With $k=200$ and $m=32$, 
the \randgreedi{} algorithm has $km= 6,400$ images at the root node but only $n/m= 313$ images at the leaves; thus the computation at the root node dominates in cost. On the other hand, as we decrease the branching factor (from $b=16$ to $2$), the number of images ($kb$) in the interior nodes decreases from $3,200$ to $400$ for the \greedyl{} algorithm. This gradual decrease in compute time is reflected in the total time and in the observed speedup.

Finally, in Fig.~\ref{fig:images} (Appendix~\ref{app:omit-figs}), we show $16$ 
out of the $200$ images determined to be cluster centers by the \greedyl\ and \randgreedi\ algorithms. We can conclude that the submodular $k$-medoid function can generate a diverse set of exemplar images for this clustering problem. 
\section{Conclusion and Future work}
\begin{comment}
    Conc:
    1. alpha/L approx
    2. Complexity for k-cover and k-medoid
    3. Reduced memory requirement
    4. Improved parallelism by making more computation and communication more parallel
    5. Showed advantages in the computation of k-medoid?

    Future work:
    1. More experiments with other constraints
    2. Apply paradigm to alternate NP-hard classes of problems
\end{comment}

We have developed a new distributed algorithm, \greedyl{}, that enhances the existing distributed algorithm for maximizing constrained submodular function. We prove \greedyl{} is \nicefrac{$\alpha\cdot b$}{$(b+m)$}-approximate,  but empirically demonstrate that its quality is close to the best approximation algorithms for several practical problems. Our algorithm alleviates the inherent serial computation and communication bottlenecks of the \randgreedi\ algorithm while reducing memory requirements. This enables submodular maximization to be applied to massive-scale problems effectively.%We also show our algorithm aleviates the inherent serial computation and communication bottlenecks of the \randgreedi\ algorithm, and reduces the memory required thus enabling submodular maximization to be solved for massive scale problems. %Finally, We showed a significant speedup in solving the popular exemplar-based clustering problem. %Overall the new \greedyl\ algorithm is more efficient in computation, communication, memory usage, and parallelism than the \randgreedi\ algorithm in certain contexts.

In the future, we plan to conduct experiments for other hereditary constraints, such as matroid and $p$-system constraints. Another direction is to apply \greedyl{}  to closely related non-monotone and weakly submodular functions. Our experiments suggest that \greedyl{} delivers higher quality solutions than the expected approximation guarantees. One area of future work could involve investigating whether the approximation ratio can be further improved.%We also plan to explore formulating an approximation ratio that better reflects our experimental observation.
%\newpage
\bibliography{main}
%\newpage\appendix
\newpage
\appendix

\section{Accumulation Trees}
\label{app:acc}
\begin{figure}[h]
\resizebox{0.49\columnwidth}{!}{
\begin{forest}
for tree={circle,draw, thick, minimum size=4.1em, s sep = 1em, inner sep=1pt, l=6em}
    [${3,0}$
        [${2,0}$
            [${1,0}$
                [${0,0}$]
                [${0,1}$]
            ][${1,2}$
                [${0,2}$]
                [${0,3}$]
            ]
        ][${2,4}$
            [${1,4}$
                [${0,4}$]
                [${0,5}$]
            ][${1,6}$
                [${0,6}$]
                [${0,7}$]
            ]
        ]
    ]
\end{forest}
}   
\resizebox{0.49\columnwidth}{!}{
\begin{forest}
for tree={circle,draw, thick, minimum size=4.1em, s sep = 1em, inner sep=1pt, l=6em}
    [${2,0}$
        [${1,0}$
            [${0,0}$]
            [${0,1}$]
            [${0,2}$]
        ][${1,3}$
            [${0,3}$]
            [${0,4}$]
            [${0,5}$]
        ][${1,6}$
            [${0,6}$]
            [${0,7}$]
        ]
    ]
\end{forest}
}  
\resizebox{0.49\columnwidth}{!}{
\begin{forest}
for tree={circle,draw, thick, minimum size=4.1em, s sep = 1em, inner sep=1pt, l=6em}
    [${2,0}$
        [${1,0}$
            [${0,0}$]
            [${0,1}$]
            [${0,2}$]
            [${0,3}$]
        ][${1,4}$
            [${0,4}$]
            [${0,5}$]
            [${0,6}$]
            [${0,7}$]
        ]
    ]
\end{forest}
}
\resizebox{0.49\columnwidth}{!}{
\begin{forest}
for tree={circle,draw, thick, minimum size=4.1em, s sep = 1em, inner sep=1pt, l=9em}
    [${1,0}$, 
        [${0,0}$]
        [${0,1}$]
        [${0,2}$]
        [${0,3}$]
        [${0,4}$]
        [${0,5}$]
        [${0,6}$]
        [${0,7}$]
    ]
\end{forest}
}
    \caption{Accumulation tree with 8 machines and branching factors 2 (top-left), 3 (top-right), 4 (bottom-left), and 8 (bottom-right). The labels inside a node represent the identification of the node.}
    \label{fig:8tree}
\end{figure}
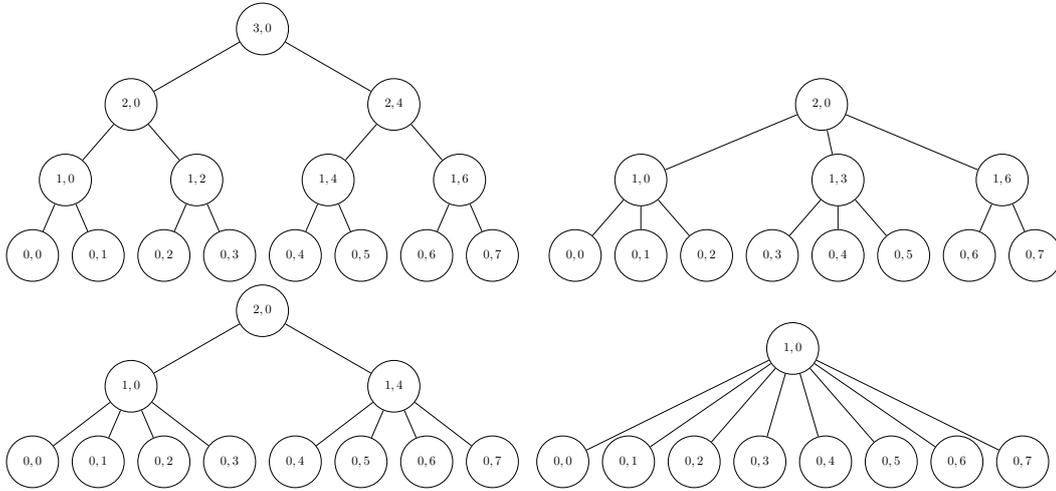
\section{Omitted Pseudocodes}
\label{app:pseudocode}
\subsection{Pseudocode of ~\randgreedi}

\begin{algorithm}[ht]
\caption{\randgreedi\ framework for maximizing constrained submodular function}
\label{Alg:randgreedi}
\begin{algorithmic}[1]
  \item[]
    { \small \Procedure{\randgreedi}{$V$: Dataset, $m$: number of machines}
    \State $S \leftarrow \emptyset$
    \State Let $\{ P_0, P_1, \ldots, P_{m-1}\}$ be an uniform random partition of $V$.
    \State \textbf{Run} $\greedy(P_i)$ on each machine $i \in [0,m-1]$ to compute the solution $S_i$  
    \State \textbf{Place} $S = \bigcup_i S_i$ on machine 0 
    \State \textbf{Run} $\greedy(S)$ to compute the solution  $T$ on machine 0
    \State \Return{$\arg\max \left\{ f(T), f(S_1), f(S_2), \dotsc, f(S_{m-1}) \right\}$}
\EndProcedure}
  \end{algorithmic}
\end{algorithm}

\subsection{Pseudocode of ~\greedyl}
Algorithm \ref{alg:multil2} describes our multilevel distributed algorithm using two procedures. The first procedure \greedyl\ is a wrapper function that sets up the environment to run the distributed algorithm. The second function $\greedyl'$ is the iterative implementation of the recurrence relation that runs on each machine.
The wrapper function partitions the data into $m$ subsets and assigns them to the machines (Line~\ref{lin:Part}). Then each machine runs the $\greedyl'$ function on the subset assigned to it (Line~\ref{lin:distRun}, Line \ref{lin:rootRun}).
The wrapper function uses and returns the solution from machine $0$ (Line~\ref{lin:rootRet}) as it is the root of the accumulation tree. \looseness=-1

The $\greedyl'$ procedure is an iterative implementation of the recurrence relation \ref{fig:RecRel} that runs on every machine. 
Each machine checks whether it needs to be active at a particular level (Line \ref{lin:chkActive})  and decides whether it needs to receive from (Line \ref{lin:rec}) or send to other machines (Line \ref{lin:send}). The function returns the solution from the last level of the machine. 
\color{green}
\begin{algorithm*}[ht]
    \caption{Our Randomized Multi-level \greedyl Algorithm }
    \label{alg:multil2}
    \begin{algorithmic}[1]
    {\small \Procedure{$\greedyl $}{$V$: Dataset, $b$: branching factor, $m$: number of machines, $r$: random tape}
    \State Let $\{P_0, P_1, \ldots P_{m-1}\}$ be uniform random partition of $V$ using $r$.\label{lin:Part}
    \For{$i=1 \ldots m-1 \text{ in parallel}$}\Comment{Run \greedyl' on all machines except $0$} 
    \State $\ell = level(i,b)$  \Comment{$level(i,b)= \max\limits_l \{l :  id \mod b^l$ is $ 0 \} $}
     \State \textbf{Run} $\greedyl'(V_i,\ell,b,i )$ to obtain $S_i$ on machine $i$  \label{lin:distRun} 
     \EndFor 
     \State \textbf{Run} $\greedyl'(V_0,\lceil \log_b m\rceil,b,0)$ to obtain $S_0$ on machine 0 \label{lin:rootRun}
     \State\Return $S_0$ \label{lin:rootRet}
   \EndProcedure}
   \end{algorithmic}
   \smallskip
   \begin{algorithmic}[1]
    {\small \Procedure{$\greedyl'$}{$P$: Partial Data-set, $\ell$: levels; $b$: branching factor, $id$: machine  ID}
    \State $S =\greedy(P)$
    \State $S_{prev}=S$
    \For{$i = 1 \dots \ell$}
        \If{$id \neq parent(id,i)$ } \label{lin:chkActive}
            \State \textbf{Send} $S_{prev}$ to $parent(id,i)$ \Comment{$parent(id, i) = b^i \cdot \lfloor id /b^i \rfloor $ } \label{lin:send}
            \State \textbf{break}
        \EndIf
        \State $D=S_{prev}$ \Comment{Prepare $D$ for current iteration}
        \For{$j = 1 \dots b-1$}{}
            \State \textbf{Receive} $D_j$ from $child(id,i,j)$ \Comment{$child(id, i, j) = id + j \cdot b^{i-1}$}  \label{lin:rec}
            \State $D = D \cup D_j$
        \EndFor
        \State \textbf{Run} $\greedy(D)$ to obtain $S$ 
        \State $S_{prev}=\arg\max\{f(S), f(S_{prev})\}$
    \EndFor
    \State \Return $S_{prev}$
    \EndProcedure}
  \end{algorithmic}
\end{algorithm*}
\color{black}

\section{Submodular Functions and Complexity}
\label{subsec:complexity} 

Our algorithm can handle any hereditary constraint, but we consider only cardinality constraints in our experiments to keep the computations simple. (More general constraints involve additional computations to check if an element can be added to the current solution set and satisfy the constraints.)
Cardinality constraints are widely used in various applications such as sensor placement~\cite{krause2008efficient}, text, image, and document summarization~\cite{lin2010multi,lin2011class}, and information gathering~\cite{krause2007near}. The problem of maximizing a submodular function under cardinality constraints can be expressed as follows.
\begin{maxi*}
{S \subseteq V}{f(S)}{} {}
\addConstraint{|S|}{\leq k}.
\end{maxi*}
Here $V$ is the ground set, $f$ is a non-negative monotone submodular function, and $k$ is the size of the solution set $S$. 

In our experiments, we have considered the following three submodular functions.
\subparagraph*{$k$-cover}
\label{func:k-cover}
Given a ground set $B$, a collection of subsets $V \subseteq 2^B$,  and an integer $k$, the goal is to select a set $S \subseteq V$ containing $k$ of these subsets to maximize $f(S)= |\bigcup_{S_i \in S}S_i |$. 
 
\subparagraph*{$k$-dominating set}
\label{func:k-dom}  The {\em $k$-dominating set} problem is a special case of the $k$-cover problem defined on graphs with the ground set $V$ as the set of vertices.  
We say a vertex $u \in V$ \emph{dominates} all its adjacent vertices (denoted by $\delta(u)$). Our goal is to select a set $S$ of $k$ vertices to dominate as many vertices as possible, i.e., $f(S) = |\bigcup_{u \in S} \delta(u) |$. The marginal gain of any vertex is the number of vertices in its neighborhood that are not yet dominated. Therefore, the problem shows diminishing marginal gains and is submodular. 
%in this graph context.  

\subparagraph*{$k$-medoid problem}
\label{func:k-med}
The {\em $k$-medoid problem}~\cite{k-medoid} is used to compute exemplar-based clustering, which asks for a set of exemplars (cluster centers) representatives of a large dataset. Given a collection of elements in a ground set $V$, and a dissimilarity measure $d(u,v)$, we define a loss function (denoted by $L$) as the average pairwise dissimilarity between the exemplars ($S$) and the elements of the data set, i.e., $L(S)= \nicefrac{1}{|V|} \sum_{u\in V} \min_{v \in S} d(u,v)$. Following~\cite{Mirzasoleiman}, we turn this loss minimization to a submodular maximization problem by setting $f(S)= L(\{e_0\}-L(S \cup \{e_0\}$, where $e_0$ is an auxiliary element specific to the dataset. The goal is to select a subset $S \subseteq V$ of size $k$  that maximizes $f(S)$. 

Next, we analyze the computational and communication complexity of our  \greedyl\ algorithm using the bulk synchronous parallel (BSP) model of parallel computation~\cite{Valiant}.
We denote the number of elements in the ground set by $n = |V|, $%\smf{V?},
the solution size by $k$, the number of machines by $m$, and the number of levels in the accumulation tree by $L$. 

\customhead{Computational Complexity}
The number of objective function calls by the sequential \greedy{} algorithm is $O(nk)$, since $k$ elements are selected to be in the solution, and we may need to compute $O(n)$ marginal gains for each of them. 
Each machine in \randgreedi{} algorithm makes $O(k(n/m + mk))$ function calls, where the second term comes from the accumulation step. Each machine of the \greedyl{} algorithm with branching factor $b$ makes $O(k(n/m + Lbk) )$  calls. Recall that $L = \lceil \log_b{m} \rceil$. 

We note that the time complexity of a function call depends on the specific function being computed.
For example,  in the $k$-coverage and the k-dominating set problems,  computing a function costs $O(\delta)$,  where $\delta$ is the size of the largest itemset for $k$-coverage,  and the maximum degree of a vertex for the vertex dominating set. 
In both cases, the runtime complexity is $O(\delta k(n/m + mk))$ for the \randgreedi,  and $O(\delta k(n/m + Lbk) )$ for the \greedyl\ algorithm. 
The $k$-medoid problem computes a local objective function value and has a complexity of $O(n^\prime \delta)$ where $\delta$ is the number of features, and $n^\prime$ is the number of elements present in the machine. 
For the leaves of the accumulation tree, $n^{\prime} = n/m$, and 
for interior nodes, $n^{\prime} = bk$. 
Therefore its complexity is $O(k \delta ((n/m)^2 + (mk)^2))$ for the \randgreedi{},  and $O(k \delta ((n/m)^2 + L(bk)^2) )$ for the \greedyl{} algorithm. 

\customhead{Communication Complexity} Each edge in the accumulation tree represents communication from a machine at a lower level to one at a higher level and contains four messages. 
They are the indices of the selected elements of size $k$, the size of the data associated with each selection (proportional to the size of each adjacency list ($ \leq \delta$), the total size of the data elements, and the data associated with each selection. 
Therefore the total volume of communication is $O(k\delta)$ per child. Since at each level, a parent node receives messages from $b$ children,
the communication complexity is $O(k\delta L b)$ for each parent. Therefore the communication complexity for the \randgreedi\ algorithm is $O(k\delta m)$ and for the \greedyl\ algorithm is $O(k\delta L \ceil{m^{1/L}}).$ 
We summarize these results in 
Table~\ref{tab:Complexity}.

\begin{table}[ht]
    \centering
    \resizebox{0.92\columnwidth}{!}{
    \begin{tabular}{l|l|ccc}
    \toprule
        Algorithms & Metric & \greedy & \randgreedi & \greedyl \\
    \midrule
    \multirow{5}{15mm}{All}
        &Elements per leaf node & $n$ & $n/m $ & $n/m$ \\
        &Calls per leaf node & $nk$ & $nk/m $ & $nk/m$ \\
        &Elements per interior node & $0$ & $km$ & $k \ceil{m^{1/L}}$ \\
        &Calls per interior node & $0$ & $k^2m$ & $k^2 \ceil{m^{1/L}}$ \\
        &Total Function Calls & $kn$ & $k(n/m + km)$ & $k(n/m + Lk \ceil{m^{1/L}})$ \\ 
        \midrule 
        \multirow{4}{20mm}{$k$-cover / k-dominating set} 
        & & \multicolumn{3}{c}{$\delta$:subset size/number of neighbours} \\ 
        & Cost Per call & $\delta$ & $\delta$ & $\delta$ \\
        & Computational complexity & $\delta kn$ & $\delta k(n/m + km)$ & $\delta k(n/m + Lk \ceil{m^{1/L}})$\\
        & Communication cost & 0 & $\delta km$ & $\delta kL \ceil{m^{1/L}}$\\
        \midrule
        \multirow{5}{15mm}{$k$-medoid } 
        & & \multicolumn{3}{c}{$\delta$: number of features} \\
        & Cost Per call in Leaf node & $ \delta n$ & $\delta n/m $ & $\delta n/m $ \\ 
        &Cost Per call in interior node & 0 & $\delta km$ & $\delta k \ceil{m^{1/L}}$\\ 
        & Computational complexity & $\delta kn^2$ & $\delta k((n/m)^2 + (km)^2)$ & $\delta k((n/m)^2 + L(k \ceil{m^{1/L}})^2)$\\
        & Communication cost & 0 & $\delta km$ & $\delta kL \ceil{m^{1/L}}$ \\
        \bottomrule
    \end{tabular}}
    \caption{Complexity Results of the submodular functions for different algorithms. The number of elements in the ground set is $n$, the selection size is $k$, 
the number of machines is $m$, and the number of levels in the accumulation tree is $L$. }
    \label{tab:Complexity}
\end{table} 

\begin{comment}
\begin{table}[ht]
    \centering
    \resizebox{\columnwidth}{!}{
    \begin{tabular}{c|c||c|c}
    \toprule
        Parameter & Description & Parameter & Description \\ \midrule
        $\alpha$ & Approximation ratio of the \greedy algorithm &
        $b$ & Branching factor of the accumulation tree \\
        \multirow{2}{*}{$m$} & Number of leaves in the accumulation tree &
        $L$ & Number of levels of the accumulation tree \\
        &Numbers of machines used for computation. &        
        $\ell$ & level identifier for a node \\
        $id$ & Machine identifier for a node. &
        $W$ & Complete universe for the input dataset \\
        $V$ & Input Dataset &
        $V_c$ & Dataset corresponding to any node $c$ of the tree \\
        $S$ & Solution Set &
         $k$ & Size of solution \\
        $P_{id}$ & Part of the dataset assigned to machine $id$ &
        $f$ & Submodular function \\
        $\fhat$ & \lsz extension of function $f$ & 
        $OPT$ & An optimal solution\\
        \bottomrule
    \end{tabular}
    }
    \caption{Notations and parameters used in the paper. }
    \label{tab:parameters}
\end{table}
    
\end{comment}

\begin{table}[ht]
    \centering
    \resizebox{\columnwidth}{!}{
    \begin{tabular}{c|c||c|c}
    \toprule
        Parameter & Description & Parameter & Description \\ \midrule
        $\alpha$ & Approximation ratio of the \greedy algorithm &
        $W$ & Complete universe for the input dataset \\
        $b$ & Branching factor of the accumulation tree &
        $w$ & Size of the universe $W$ \\
        \multirow{2}{*}{$m$} & Number of leaves in the accumulation tree & $V$ & Input Dataset \\
        &Numbers of machines used for computation. &
        $n$ & Size of the input dataset $V$ \\
        $L$ & Number of levels of the accumulation tree &
        $V_c$ & Dataset corresponding to any node $c$ of the tree \\ 
        $\ell$ & level identifier for a node & 
        $S$ & Solution Set \\
        $id$ & Machine identifier for a node. &
        $k$ & Size of solution \\
        $f$ & Submodular function  & 
        $OPT$ & The optimal solution \\
        $\fhat$ & \lsz extension of function $f$ &
        \multirow{2}{*}{$p(e)$} & Probability that e $\in$ OPT is selected\\
        $P_{id}$ & Part of the dataset assigned to machine $id$ &
        &  by the \greedy algorithm when sampled from V. \\
        \bottomrule
    \end{tabular}
    }
    \caption{Notations and parameters used in the paper. }
    \label{tab:parameters}
\end{table}

\section{Omitted Proofs}
\label{app:omit-proofs}

\subsection{Proof of Lemma~\ref{lem:individual}}
\label{lem:individual:proof}
\begin{proof}
We first construct a subset of $OPT$ that contains all the elements that do not appear in $S_c$ when added to some leaf node in the subtree rooted at child $c$.
Let $O_c$ be the rejected set that can be added to $V_c$ without changing $S_c$; i.e., 
$O_c =  \{ e \in OPT :  e \notin \greedy(V_c \cup \{e\})\}.$
Therefore, $\Pr[e \in O_c] = 1 - \Pr[e \notin O_c] = 1- p(e). $

From Lemma~\ref{lem:stable}, we know that $\greedy(V_c \cup O_c) = \greedy(V_c) $ .  
Since the rejected set $O_c \subseteq OPT$ and the constraints are hereditary, $O_c \in \mathcal{C}$ (i.e $O_c$ is a feasible solution of child node c).
Then from the condition of Lemma \ref{lem:individual}, we have
\begin{align*}
f(S_c) &\geq \alpha \cdot f(O_c) \\
    \mathbb{E}[f(S_c)] &\geq \alpha \cdot \mathbb{E}[f(O_c)] = \alpha \cdot \fhat(\mathbb{E}_{V_n}[1_{O_c}]) = \alpha\cdot \fhat(1_{OPT_{\ell,id}} - p_{\ell,id}). \notag 
\end{align*}
\end{proof}

\subsection{Proof of Lemma~\ref{lem:accumulate}}
\label{lem:accumulate:proof}
\begin{proof}
We first show a preliminary result on the union set $D$. Consider an element $e\in D\cap OPT$ present in some solution $S_c$ from a child $c$. Then, 
$$\Pr[e\in S_c | e \in V_c]
    = \Pr[e\in \greedy(V_c) | e \in V_c]. $$

Since the distribution of $V_c \sim V(1/m)$ conditioned on $e \in V_c$ is identical to the distribution of $B \cup \{e\}$, where $B \sim V(1/m)$,  we have,
$$
\Pr[e\in S_c | e \in V_c] = \Pr_{B \sim V(1/m)}[ e \in \greedy(B \cup \{e\})]  = p(e). $$

Since this result holds for every child $c$, and each subset $V_c$  is disjoint from the corresponding subsets mapped to the other children,  we have 
\begin{align*}
    \Pr(e \in D\cap OPT) &= \sum_{i}\Pr[e\in S_{c_i} \cap OPT | e \in V_{c_i}]\Pr[e\in V_{c_i}]. &= \sum_i p * 1/m = bp/m.
\end{align*} 
 Now, we are ready to prove the Lemma. 
 The subset $D \cap OPT_{\ell, id} \in \mathcal{C}$, since it is a subset of $OPT_{\ell, id}$ and the constraints are hereditary. Further, since the \greedy\ algorithm is $\alpha$-approximate, we have 
\begin{align}
    f(S) &\geq \alpha \cdot f(D\cap OPT_{\ell,id}) \notag \\
\mathbb{E}_{V_n}[f(S)] &\geq \mathbb{E}_{V_n}[\alpha\cdot f(D\cap OPT)] \notag \\
    &\geq \alpha\cdot \fhat(\mathbb{E}_{V_n}[1_{D\cap OPT}]) 
    && \pushright{[\text{\lsz Ext.~(2)}, \ref{eqn:Lova}}]  %\ref{eqn:Lov2}})}
    \notag \\
    &= \alpha\cdot \fhat(bp/m) = \alpha b/m \cdot \fhat(p). 
    &&\pushright{[\text{\lsz Ext.~(3)} \ref{eqn:Lova}]}  \label{eqn:lemma3} 
\end{align} 
\end{proof}

\section{Omitted Results}
\label{app:omit-figs}
\begin{figure}[ht]
    \centering
    \resizebox{0.9\columnwidth}{!}{
    \begin{subfigure}[t]{0.49\textwidth}
        \includegraphics[width=0.8\textwidth]{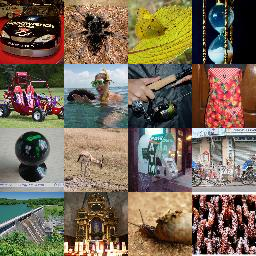}
    \end{subfigure}
    \begin{subfigure}[t]{0.49\textwidth}
        \includegraphics[width=0.8\textwidth]{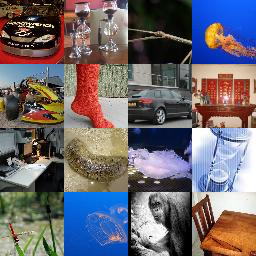}
    \end{subfigure}}
    \caption{Results from \greedyl{} for the $k$-medoid problem on the Tiny ImageNet dataset on 32 nodes with $k=200$ with no images added at each accumulation step. The subfigure on the left shows the first 16 image results for one of the runs for the \greedyl{} algorithm with branching factor $b=2$, and the subfigure on the right shows the top 16 image results for one of the runs for the \randgreedi{} algorithm.}
    \label{fig:images}
\end{figure}

\label{app:omit-tab}
\begin{table}[htb]
    \centering
    \resizebox{0.6\columnwidth}{!}{
    \begin{tabular}{|c|c|c|} \hline
        Branching Factor & First accumulation & Final accumulation \\ \hline
        4  &0.74111	&0.99994 \\
        8  &0.82971	&1.00005 \\
        16 &0.91965	&0.99994 \\
        32 &1.00003	&1.00003 \\ \hline

    \end{tabular}}
    \caption{Objective function values relative to the \greedy algorithm at the first and final accumulation steps for Friendster with selection size $k = 1000$ and $m = 32$.}
    \label{tab:quality-levels}
\end{table}

\end{document}